\documentclass[11pt]{article}

\usepackage{amssymb,amsmath,latexsym,amscd,amsfonts,bbold,url,algorithm,algpseudocode}

\usepackage{graphicx}  
\usepackage[usenames, pdftex, dvipsnames]{color}

\usepackage[notref,notcite,final]{showkeys}
\usepackage{umoline}
  
\newtheorem{theorem}{Theorem}[section]

\newtheorem{claim}[theorem]{Claim}

\newtheorem{corol}[theorem]{Corollary}

 \newcommand{\qedsymb}{\hfill{\rule{2mm}{2mm}}}  
 \newenvironment{proof}[1][]{\begin{trivlist}  
 \item[\hspace{\labelsep}{\bf\noindent Proof#1:\/}] 
 }{\qedsymb\end{trivlist}}

\newcommand{\ignore}[1]{}

\newcommand{\QMA}{\mathsf{QMA}}

\newcommand{\NP}{\mathsf{NP}}

\newcommand{\FIX}{\mathsf{FIX}}

\newcommand{\ket}[1]{{ |{#1} \rangle }}  
\newcommand{\bra}[1]{{ \langle {#1} | }}

\newcommand{\poly}{\mathrm{poly}} 
\newcommand{\EqDef}{\stackrel{\mathrm{def}}{=}}
\newcommand{\Tr}{\mathrm{Tr}}
\newcommand{\pr}{\text{Pr}}

\newcommand{\Eq}[1]{Eq.~(\ref{#1})}

\newcommand{\Cor}[1]{Corollary~\ref{#1}}

\newcommand{\Ref}[1]{Ref.~\cite{#1}}

\newcommand{\Thm}[1]{Theorem~\ref{#1}}
\newcommand{\App}[1]{Appendix~\ref{#1}}

\newcommand{\Id}{\mathbb{1}}

\newcommand{\bs}{\bar{s}}

\mathchardef\mhyphen="2D 
\newcommand{\qsat}{\textsc{qsat}}
\newcommand{\cqsat}{\textsc{commuting qsat}}
\newcommand{\kqsat}{k\mhyphen\textsc{qsat}}
\newcommand{\sat}{\textsc{sat}}
\newcommand{\ksat}{k\mhyphen\textsc{sat}}

\newcommand{\clh}{\textsc{commuting local hamiltonian}}

\begin{document}
\author{Or Sattath\footnote{School of Computer Science and Engineering, The Hebrew University, Israel.} \and Itai Arad\footnote{Centre for Quantum Technologies, National University of Singapore, Singapore.}}

\title{A Constructive Quantum Lov\'asz Local Lemma for Commuting
  Projectors} 

\maketitle

\begin{abstract}
  The Quantum Satisfiability problem generalizes the Boolean
  satisfiability problem to the quantum setting by replacing
  classical clauses with local projectors. The Quantum Lov\'asz
  Local Lemma gives a sufficient condition for a Quantum
  Satisfiability problem to be
  satisfiable~\cite{ambainis2012quantum}, by generalizing the
  classical Lov\'asz Local Lemma.

  The next natural question that arises is: can a satisfying quantum
  state be \emph{efficiently} found, when these conditions hold?
  In this work we present such an algorithm, with the additional
  requirement that all the projectors commute. The proof follows
  the information theoretic proof given by Moser's breakthrough
  result in the classical setting~\cite{moser2009constructiveSTOC}.

  Similar results were independently published
  in~\cite{cubitt2011constructive, cubitt2013quantum}.
\end{abstract}

\section{Introduction and main results} 

Given a set of independent events, $B_{1},\ldots,B_{m}$ such that
$\pr(B_{i}) < 1$, the probability that none of the events happen is
strictly positive: \[ \pr(\bigwedge_{i=1}^{m} \overline{B_{i}}) =
\prod_{i=1}^{m}( 1-\pr(B_{i})) > 0.\] What if the events are not
independent? The Lov\'asz Local Lemma (LLL) provides a definition of
``weakly dependent'' events, and is an important tool to guarantee
that the probability that none of the events happens is strictly
positive.  
\begin{theorem}[\cite{erdos1975problems}, 
  see also Ref.{~\cite[Chapter 5]{alon2004probabilistic}}] 
\label{th:symmetric_lll} 
  Let $B_1, B_2,\ldots, B_m$ be events with $\pr(B_i)\leq p$ and
  such that each event is mutually independent of all but $g-1$
  events.  If $p \cdot e \cdot g\le 1$ then
  $\pr(\bigwedge_{i=1}^m \overline{B_i})>0$.
\end{theorem}
One of the many uses of the LLL is providing a sufficient condition
for the satisfiability of a $\ksat$ formula.
\begin{corol}[\cite{kratochvil1993one}]
\label{cor:lll_for_clauses} 
  Let $\Phi$ be a $\ksat$ formula in CNF-form in which each clause
  shares variables with at most $g-1$ other clauses.  Then if $g\le
  \frac{2^k}{e}$, $\Phi$ is satisfiable.
\end{corol} 
The corollary follows from Thm.~\ref{th:symmetric_lll} by letting
$B_i$ be the event that the $i$-th clause is not satisfied for a
random assignment, which happens with probability $p=2^{-k}$, and
noting that by definition, every such event is independent of all but
at most $g-1$ other events.

In \Ref{ambainis2012quantum}, a quantum version of
\Thm{th:symmetric_lll} was proved, known as the \emph{Quantum
Lov\'asz Local Lemma (QLLL)}. In the quantum settings, one replaces
the probability space by the quantum Hilbert space, events by
subspaces and their probabilities by the ``relative dimension'',
which is the ratio of the subspace dimension to the total dimension
of the Hilbert space. The formal statement of this theorem, as well
as the exact definitions, are not relevant for this work, and
therefore will be omitted here; these details can be found in 
\Ref{ambainis2012quantum}.

The QLLL also has a simple corollary, which is equivalent to
Corollary~\ref{cor:lll_for_clauses}. To state it, we first need to
define the The \emph{quantum satisfiability problem}, $\kqsat$,
which is the quantum analog of the classical $\ksat$ problem.
Introduced by Bravyi in~\Ref{bravyi2006efficient}, $\kqsat$ is the
following promise problem. We are given a set of projectors
$\{\Pi_{1}, \ldots, \Pi_{m}\}$ that are defined on the Hilbert space
of $n$ qubits. Each projector $\Pi_i$ acts non-trivially on at most
$k$ qubits, which means that it can be written as $\Pi_i =
\hat{\Pi}_i \otimes \Id_{n-k}$, where $\hat{\Pi}_i$ is a projection
defined in the Hilbert space of $k$ qubits, and $\Id_{n-k}$ is the
identity operator on the rest of the qubits. In addition, we are
given a parameter $\epsilon = \Omega(1/\poly(n))$, as well as a
promise that either there exists a quantum state $\ket{\psi}$ in the
intersection of the null spaces of all the projectors (in which
case we say that the instance is \emph{satisfiable}), or that every
state $\ket{\psi}$ satisfies
$\sum_{i=1}^{m}\bra{\psi}\Pi_{i}\ket{\psi} \geq \epsilon$.  Our goal
is to decide which of these possibilities holds. We note that just
as in the classical case, where $\ksat$ is 
$\NP$-complete for $k \geq 3$, $\kqsat$ is 
\emph{quantum} $\NP$-complete
for $k \geq 3$~\cite{gosset2013quantum}.\footnote{More precisely, it is
complete for the class $\QMA_1$, which is
the usual quantum $\NP$ class $\QMA$, but with a one-sided error
(see \Ref{bravyi2006efficient} for details).}

Next, we say that the \emph{neighborhood} of a projector $\Pi_{i}$,
denoted $\Gamma^{+}(\Pi_{i})$, is the set of projectors that act
nontrivially on at least one of the qubits on which $\Pi_{i}$ acts 
nontrivially (note that $\Pi_{i} \in \Gamma^{+}(\Pi_{i}))$.  We say
that the problem is a ($k$-local, rank-$r$,
$g$-neighborhood)-$\qsat$ if it is a $\kqsat$ problem with the
additional properties that each projector is of rank at most $r$,
and $|\Gamma^+(\Pi_i)|\le g$ for every projector $\Pi_{i}$. In terms
of these definitions, we have the following corollary of the QLLL:
\begin{corol}[\cite{ambainis2012quantum}]
  \label{cor:qsat_satisfiable_clauses} Let $\{\Pi_1,\ldots,\Pi_m\}$
  be a ($k$-local, rank-$r$, $g$-neighborhood)-$\qsat$ instance. If
  $g\le\frac{2^k}{re}$, the instance is satisfiable.
\end{corol}
A rank-1 $\kqsat$ instance is similar to a $\ksat$ instance in the
following sense: in a $k$-CNF formula, each clause excludes one out
of $2^{k}$ configurations of the relevant variables; in a rank-1
$\kqsat$, each projector excludes one dimension out of the $2^{k}$ relevant
dimensions. The parameters of \Cor{cor:qsat_satisfiable_clauses} and
\Cor{cor:lll_for_clauses} coincide in the case of rank-1 instances.
We keep the generality of rank-r instances for reasons that are
related to commuting instances, to be clarified
shortly.

%

\Cor{cor:lll_for_clauses} guarantees the satisfiability of a $\sat$
instance under a certain condition. How difficult is it to find a
satisfying assignment under this condition, in terms of
computational resources? Efficient constructive versions of the LLL
started with the work of Beck~\cite{beck_algorithmic_1991}, which
provided an algorithm that worked under stronger conditions than the
LLL. The results were improved by others, and culminated in the work
by Moser and Tardos~\cite{moser_constructive_2009}(see the
references therein for the complete line of research), which
provided an efficient algorithm under the same conditions as in the
LLL.

Our main result is Algorithm~\ref{alg:compression_algorithm}, an
efficient quantum algorithm for a \emph{commuting} {$\kqsat$}
instance, which satisfies the conditions of the QLLL, i.e.,
\Cor{cor:qsat_satisfiable_clauses}. The $\cqsat$ problem adds the
following requirement to $\qsat$: $[\Pi_{i},\Pi_{j}] = 0$ for all
$i,j$. The \cqsat\ (and more generally, the \clh\ problem) is an
intermediate problem between the classical and quantum regime. On
one hand, all the terms can be diagonalized simultaneously (like the
classical setting); on the other hand, this basis may be highly
entangled, and possesses rich non-classical phenomena (for example,
Kitaev's toric code~\cite{kitaev2002classsical}). Various results
concerning the complexity of this problem appeared
in~\cite{bravyi2005commutative,aharonov2011complexity,schuch2011complexity}.
Our proof for the correctness of the algorithm is based on an
information theoretic argument, which is similar to the one used in
Moser's information theoretic proof (see~\cite{fortnow_kolmogorov}).
\begin{theorem}[Main result]
\label{thm:main} Let $\{\Pi_1,\ldots,\Pi_m\}$ be a ($k$-local, rank
  $r$, $g$-neighborhood)-$\qsat$ instance with $g < \frac{2^k}{re}$.
  Then for every $\delta > 0$, 
  Algorithm~\ref{alg:compression_algorithm} returns ``Success'' with
  probability $\ge 1-\delta$ and a running time of
  $m\cdot\tilde{\mathcal{O}}(\eta)$, where $\eta\EqDef
  \frac{1}{\delta[k-\log(ger)]}$.
 
  If $\forall i,j \leq m \ [\Pi_i,\Pi_j] = 0$ then the output is a
  satisfying state.
\end{theorem}

\begin{algorithm}[!hbtp]
  \caption{Commuting QLLL solver}\label{alg:compression_algorithm}
    
    {~}
    
  	\textbf{Initialize:}
    \begin{algorithmic}[1]
      \State For $\eta\EqDef \frac{1}{\delta[k-\log(ger)]}$, fixed integer
      \begin{align}
        \label{def:T}
          T \EqDef \lceil 4m\eta\cdot \log(\eta+2)\rceil \ .
      \end{align}
      
      \State System register: $n$ qubits, prepared in a fully mixed state.
      \State Stock register: $N=Tk$ qubits in fully mixed state.
    \end{algorithmic}

    {~}
    
  	\textbf{Algorithm:}
    \begin{algorithmic}[1]
    \State $t \gets 0$
    \For{$i \gets 1$ to $m$}
        \State Fix$(\Pi_i)$ \label{alg:fix_all}
    \EndFor
    \State return ``Success''
    
    {~}

    \Procedure{\textnormal{Fix}}{$\Pi$}
      \State measure $\{\Pi, \Id - \Pi\}$; \label{alg:measurement} in 
        the system register \label{alg:measurement_step}
      \If{result $=\Pi$ (projector violated, energy = 1)}
        \State $t\gets t+1$
        \If{$t=T$}
				   \State abort and return ``Failure''
				\EndIf
        \State Replace the $k$ measured qubits with $k$ maximally
          mixed qubits from the stock register
        \ForAll{$\Pi_j \in \Gamma^+(\Pi)$}
          \State Fix$(\Pi_j)$ \label{alg:call_to_fix}
        \EndFor
      \EndIf
    \EndProcedure
  \end{algorithmic}
\end{algorithm}

Very similar variants of the main result were discovered
independently. A talk that described the different approaches was
given in~\Ref{cubitt2011three}. The main open question emerging from
these works is how to find a constructive version for the QLLL in
the non-commuting case. We hope that at least one of these
approaches turns out to be useful for proving the non-commuting
case. We now compare the differences between these approaches. We
believe that the main (and only) advantage of our approach is its
simplicity. The result given in ~\Ref{cubitt2011constructive} holds
in a more general setting which is called the \emph{asymmetric}
QLLL, whereas our version only holds in the so-called
\emph{symmetric} QLLL. The asymmetric LLL (see, e.g.
Ref.~\cite[Lemma 5.1.1]{alon2004probabilistic}) is useful when
there are differences between the upper bounds on the probabilities
of the events. Furthermore, in the non-commuting case, the
termination of the algorithm in \Ref{cubitt2011constructive} implies
that the state has low energy (so, the remaining task to prove the
non-commuting version is to prove a fast termination). On the other
hand, our algorithm terminates also in the non-commuting case, and
the missing part is proving that the state has low-energy. These
results are complementary in that sense. The approach
in~\Ref{cubitt2013quantum} has a much better trade-off between the
running time and the probability of success.
 
Before proving \Thm{thm:main}, let us briefly explain why we need to
work in the general setting with rank-$r$ projectors, and not only
with rank-$1$. It is easy to verify that a ($k$-local, rank
$r$)-$\qsat$ is equivalent to a ($k$-local, rank $1$)-$\qsat$ by
replacing each rank-$r$ projector with $r$ different rank-1
projectors. This is because a rank-$r$ projector can be written as
the sum of $r$ rank-1 projectors. Unfortunately, this transformation
breaks the commutativity property for $\cqsat$; for example, it is possible 
that $[A,B]=0$ for some projectors $A$ and $B$, but for the
decomposition to rank-1 projectors $A = \sum_i A_i$, $B=\sum_j B_j$, $[A_i,B_j] \neq 0$
for some $i$ and $j$.

\section{Analyzing the algorithm -- proof of \Thm{thm:main}}

In the commuting case, if the algorithm succeeds (i.e., returns
``Success''), it produces a satisfying state $\rho$: the set of
satisfied projectors monotonically increases when a $\FIX$ call
returns, and furthermore, after $\FIX(\Pi_{i})$ returns,
$\Tr(\Pi_{i} \rho) = 0$. Because $\FIX$ is called for every
projector (line \ref{alg:fix_all}), if the algorithm succeeds, the
resulting state on the system register satisfies all the
projectors.  

Therefore, in order to prove \Thm{thm:main}, we need to show that
the success probability is at least $1-\delta$. We analyze the
success probability of the algorithm by deriving an inequality that
relates the initial entropy of the system to that of the
different possible branches in the running history of the algorithm.

Initially, the two registers are completely mixed, so the system
is in the state $\rho_{init} =
2^{-(n+N)}\Id_{n+N}$, and
\begin{align*}
  S(\rho_{init}) = n+N \ .
\end{align*}

{~}

To analyze the final state of the system, we note that all possible
sequences of measurements form a history tree. We characterize each
history branch by a binary string that records the result of the
measurements in the branch. 0 means a ``success'' --- a projection
into the zero energy subspace of the projector, and 1 means a
``failure''. We denote such string by $\bs=(s_1, s_2, s_3, \ldots)$,
where $s_i$ denotes the outcome of the $i$th measurement, and let
$|\bs|$ denote the Hamming weight of $\bs$, which is exactly the
number of failures in a particular branch. Denote by $\rho_{\bs}$
the normalized resulting state of such a series of projective
measurements, and $p_{\bs}$ the probability of that branch to occur.

The heart of the analysis is the following simple claim, that
upper-bounds the entropy of the initial state by an expression involving
$\{p_{\bs}\}$ and $\{\rho_{\bs}\}$.

\begin{claim}
\label{fa:entropy} Consider an (adaptive) quantum algorithm that
  applies a series of projective measurements to an initial state
  $\rho_{init}$, and let $O$ denote the set of all possible
  measurements outcomes. For each $s\in O$, let $p_s$ denote the
  probability with which outcome $s$ occurs, and $\rho_s$ the
  resultant final state of the system. Then 
  \begin{align*}
    S(\rho_{init}) \le  H(\{p_{s}\}) 
      + \sum_{s \in O}p_{s}S(\rho_{s}) \ ,  
  \end{align*}
  where $H(\cdot)$ is the Shanon entropy of classical probability
  distributions, and $S(\cdot)$ is the von Neumann entropy.
\end{claim}

\begin{proof}  
  By adding a sufficient amount of ancilla qubits to the system,
  initialized in the pure state $\ket{0}$, we can use standard
  techniques from the theory of quantum computation, and move all
  intermediate measurements to the end of the algorithm, where they
  will be performed on the ancilla qubits. Specifically, if $L$ is
  the number of ancilla qubits we add, then we can assume without
  loss of generality that we start with the state 
  \begin{align}
    \rho'_{init} \EqDef \rho_{init}\otimes \ket{0^{\otimes
        L}}\bra{0^{\otimes L}} \ ,
  \end{align}
  apply to it some unitary circuit $U$, obtain the state
  \begin{align}
    \rho'_1 \EqDef U^\dagger \rho'_{init} U \ ,
  \end{align}
  and in the end perform a
  projective measurement on the ancilla qubits. The measurement is
  defined by the set of orthogonal projectors $\{P_s\}_{s\in O}$,
  such that $\sum_{s\in O} P_s =\Id$, and we are guaranteed that
  \begin{align*}
    p_s  &= \Tr \big(P_s \rho'_1\big) \ , \\
    \rho_s &= \frac{1}{p_s}  
      \Tr_{ancilla} \big(P_s \rho'_1 P_s\big) \ .
  \end{align*}
  
  Let us now define the state $\rho'_{final} \EqDef
  \sum_{s\in O} P_s \rho'_1 P_s$, which is the state of the
  system after we measured the ancillary qubits, and let us define
   $\rho'_s \EqDef \frac{1}{p_s} P_s \rho'_1
  P_s$.
  Clearly, $\rho'_{final} = \sum_s p_s \rho'_s$, and $\rho'_s$ have
  orthogonal supports. We now apply the following 
  two elementary results from the theory of quantum
  information:
  \begin{enumerate}
    \item A projective measurement can only increase the von Neumann
      entropy (see Nielsen \& Chuang Theorem~11.9):
      \begin{align*}
        S(\rho'_1) \le S(\rho'_{final}) \ .
      \end{align*}
     
    \item The entropy of a sum of mixed states with orthogonal support (see
      Nielsen \& Chuang Theorem~11.8(4)):
      \begin{align*}
        S(\sum_s p_s\rho'_s) = H(\{p_s\}) + \sum_s p_s S(\rho'_s) \ .
      \end{align*}
      
  \end{enumerate}
  Combining these two results, together with the easy observation that
  $S(\rho_{init}) = S(\rho'_{init}) = S(\rho'_1)$, we conclude that
  $S(\rho_{init}) \le H(\{p_s\}) + \sum_s p_s S(\rho'_s)$. Finally, $\rho'_s
  = \rho_s \otimes \ket{s}\bra{s}$. Therefore $S(\rho'_s) =
  S(\rho_s)$, and this complete the proof.
  
\end{proof}

We now return to the analysis of the algorithm. Applying 
Claim~\ref{fa:entropy}, we get
\begin{align}
\label{eq:main}
  n  +N = S(\rho_{init}) 
    \le H\big(\{p_{\bs}\}) + \sum_{\bs} p_{\bs} S(\rho_{\bs}) \ .
\end{align}

We define $p_t$ as the probability that
Algorithm~\ref{alg:compression_algorithm} ended with exactly $t$
failures:
\begin{align*}
  p_{t} \EqDef  \sum_{|\bs|=t} p_{\bs}.
\end{align*}

Let us upper bound the RHS of \Eq{eq:main} in terms of the probabilities
$\{p_t\}$.
\begin{description}
  \item [1. $H(\{p_{\bs}\})$:]\  \\
    By definition, $H(\{p_{\bs}\})= H(\{p_{t}\}) + \sum_{t=0}^{T}
    p_{t} H\big(\{p_{\bs} \}\big| \, |\bs| = t\big)$. Since $t$ can
    take values between 0 and $T$ then trivially $H\big(\{ p_t
    \}\big) \le \log T$. 
    
    To upper-bound $H\big(\{p_{\bs}\}\big| |\bs| = t\big)$, we count
    the number of strings with exactly $t$ ``1''. Every such string
    corresponds to a branch with exactly $t$ failures.  The total
    length of the string is the total number of measurements, which
    is the total number of calls to $\FIX$. This is at most $m+gt$
    because we had at most $m$ external calls (we had exactly $m$
    such calls when $t<T$, but for $t=T$ we could have had less) and
    every failure triggers at most $g$ recursive calls. Therefore,
    \begin{align*}
      H\big(\{\hat{p}_{\bs}\} \big| |\bs| = t \big)
      &\le \log\binom{m+gt}{t} \\
      &\le m + \log\binom{gt}{t} \\
      &\le m + \log \left(\frac{egt}{t}\right)^t
        = m +  t\log(ge) \ .
    \end{align*}
    The second inequality is valid for $g\ge 2, t\ge 1$ and $m\ge
    0$, and is proved in \App{app:binom-upperbound}. The second
    inequality follows from the standard bound
    $\binom{n}{k} \le \left(\frac{en}{k}\right)^k$. All together,
    \begin{align*}
      H(\{p_{\bs}\}) \le \log T + \sum_{t=0}^T 
        p_t\cdot\big(m + t\log(ge)\big) \ .
    \end{align*}
    
    \item [$\sum_{\bs} p_{\bs} S(\rho_{\bs})$:] \ \\
      We write, $\sum_{\bs} p_{\bs} S(\rho_{\bs}) = \sum_{t=0}^T
      \sum_{\bs, |\bs|=t} p_{\bs} S(\rho_{\bs})$. Each branch in the
      internal sum has exactly $t$ failures, and each such failure
      is a measurement which puts $k$ qubits in an $r$-dimensional
      subspace (because the $\Pi_i$ are rank-$r$ projectors).
      The entropy of these qubits is therefore
      at most $k-\log r$, and so for such branch, $S(\rho_{\bs})
      \le N + n - t(k - \log r) $. Therefore,
      \begin{align*}
        \sum_{\bs} p_{\bs} S(\rho_{\bs}) 
          &\le \sum_{t=0}^T \sum_{\bs , |\bs|=t} p_{\bs}
          \big(N + n - t(k - \log r ) \big) \\
          &= N + n - (k - \log r )\sum_{t=0}^T p_t\cdot t  \ .
      \end{align*}
  \end{description}

Plugging these bounds into \Eq{eq:main}, we get
\begin{align*}
  n + N \le \log T + N + n + m 
    - [k-\log(ger)]\sum_{t=0}^Tp_{t}\cdot t \ ,
\end{align*}
which implies  $[k-\log(ger)]\sum_{t=0}^Tp_{t}\cdot t \le \log T +
m$. By assumption, $k-\log(ger) > 0$, and therefore,
\begin{align}
\label{eq:failure}
  \Pr(\text{failure}) = p_T \le \frac{\log T + m}{T[k-\log(ger)]} \ .
\end{align}
Recalling that $T \EqDef \lceil 4m\eta\cdot \log(\eta+2)\rceil$,
it is now a simple algebra
to show that when $m\ge 2$, the RHS of the above equation is upper
bounded by $\delta$. See \App{app:RHS} for a full proof. This finishes
the proof of \Thm{thm:main}.

\section*{Acknowledgments}
The authors wish to thank Toby Cubitt, Julia Kempe, Frank Verstraete and especially Martin Schwarz for valuable discussions.

Research at the Centre for Quantum Technologies is funded by the
Singapore Ministry of Education and the National Research
Foundation, also through the Tier 3 Grant ``Random numbers from
quantum processes".

IA and OS received funding from the European Research Council under the European Union's Seventh Framework Programme (FP7/2007-2013) / ERC grant agreement no.\ 280157. OS was supported by Julia Kempe's Individual Research Grant of the Israeli Science Foundation and by Julia Kempe's European Research Council (ERC) Starting Grant.
 
OS would like to thank the Clore foundation for their support.

\appendix

\section{Proofs}
\label{sec:proofs}

\subsection{An upper bound on $\binom{m+gt}{t}$}
\label{app:binom-upperbound}

Let $g\ge 2, t\ge 1$ be some integers. We would like to show that for every
$m\ge 0$,
\begin{align}
\label{eq:binom-ineq}
  \binom{m+gt}{t} \le 2^m \binom{gt}{t} \ .
\end{align}
The proof is by induction. For $m=0$, the inequality in
\eqref{eq:binom-ineq} is trivially satisfied. Assume then that it is
true for $m>0$, and let us prove its validity for $m+1$. Writing
$\binom{m+gt}{t} = \frac{(M+gt)!}{(m+gt-t)!t!}$, one can easily
verify that
\begin{align*}
  \binom{m+gt}{t} = \frac{m+gt}{m+gt-t}\binom{m-1+gt}{t} \ .
\end{align*}
Therefore, by the induction assumption, $\binom{m+gt}{t}\le
\frac{m+gt}{m+gt-t} 2^{m-1}$. Finally, using the fact that $g\ge 2,
t\ge 1$ and $m\ge 0$, it is easy to see that $\frac{m+gt}{m+gt-t}\le
2$, and therefore $\binom{m+gt}{t}\le 2^m$.

\subsection{Bounding the RHS of \Eq{eq:failure}}
\label{app:RHS}

To show that the RHS of \Eq{eq:failure} is upperbounded by $\delta$,
we need to show that 
\begin{align}
\label{eq:T-ineq}
  \frac{\log T + m}{T} \le \delta[k-\log(ger)] \EqDef \frac{1}{\eta}
  \ .
\end{align}
Since the LHS of the above inequality is a decreasing function of
$T$ for every $T\ge 1$ and $m\ge 2$, we can safely drop the
$\lceil\cdot\rceil$ from the definition of $T$, and prove the
inequality for $T= 4m\eta\cdot\log(\eta+2)$. Substituting this in
\Eq{eq:T-ineq}, we obtain the following equivalent inequality
\begin{align*}
  \frac{\log(4m) + \log\eta + \log\log(\eta+2) + m}{m\log(\eta+2)}
  \le 4 \ .
\end{align*}
It is now straight forward to verify that since $\eta>0$ then as
long as $m\ge 2$, 
\begin{align*}
  \frac{\log(4m)}{m} &\le 2 \ , &
  \frac{\log\eta + \log\log(\eta+2)}{m\log(\eta+2)} &\le 1 \ , &
  \frac{m}{m\log(\eta+2)} &\le 1 \ ,
\end{align*}
and this finishes the proof.

\bibliographystyle{alpha}

{~}

\bibliography{cqlll}

\end{document}